\documentclass[10pt]{article}
\usepackage[letterpaper,margin=0.72in]{geometry}
\usepackage[T1]{fontenc}
\usepackage{lmodern}
\usepackage{microtype}
\usepackage{amsmath,amssymb,amsthm,mathtools}
\usepackage{booktabs,array}
\usepackage{float}
\usepackage{algorithm}
\usepackage{algpseudocode}
\algrenewcommand\algorithmicrequire{\textbf{Input:}}
\usepackage{tikz}
\usetikzlibrary{arrows.meta,calc,positioning,decorations.pathreplacing}
\usepackage{caption}
\usepackage[nocompress]{cite}
\usepackage[hidelinks]{hyperref}
\hypersetup{pdftitle={A 2.37332-Competitive Algorithm for Online Square Packing with Gravity},pdfauthor={Nichlas Langhoff Rasmussen (Alpha Energy ApS)},pdfsubject={Online square packing with Tetris and gravity constraints}}
\usepackage{url}

\newtheorem{lemma}{Lemma}
\newtheorem{theorem}[lemma]{Theorem}
\newtheorem{corollary}[lemma]{Corollary}
\newtheorem{proposition}[lemma]{Proposition}

\newcommand{\AsymmetricSlots}{\textnormal{\textsc{AsymmetricSlots}}}
\newcommand{\AS}{\text{\normalfont\textsc{AsymmetricSlots}}}
\newcommand{\OPT}{\operatorname{OPT}}
\newcommand{\CR}{\operatorname{CR}}
\newcommand{\Hvirt}{H_{\mathrm{virt}}}

\title{A 2.37332-Competitive Algorithm for Online Square Packing with Gravity}
\author{Nichlas Langhoff Rasmussen\thanks{Alpha Energy ApS. \href{mailto:nichlas.rasmussen@gmail.com}{\texttt{nichlas.rasmussen@gmail.com}}.}}
\date{}

\begin{document}
\maketitle
\vspace{-1.5em}

\begin{abstract}
We consider online packing of axis-parallel squares into a unit-width strip under the Tetris and gravity constraints: An incoming square must be lowered from above along a monotonic downwards path until it reaches support from below. Fekete, Kamphans, and Schweer [Algorithmica, 2014] gave an algorithm with asymptotic competitive ratio $34/13\approx2.6154$ in this model. We present \AsymmetricSlots, a recursive algorithm based on splitting each slot into a wide and narrow subslot. The proof uses a local charging argument: squares that are large relative to its associated slot pay for the height they create with their own area, while smaller squares are balanced between the two subslots and may use a bounded temporary credit. For a suitable split parameter $p^\star$, we prove
\[
\AS_{p^\star}(\sigma)\le 2.37332\,\OPT(\sigma)+O(1)
\]
for every input sequence $\sigma$. Additionally, we show that the same framework gives an algorithm with asymptotic competitive ratio $O(\kappa)$ for rectangles of aspect ratio at most $\kappa$, and a matching $\Omega(\kappa)$ lower bound shows that the dependence on $\kappa$ is asymptotically optimal. For the square algorithm, we give a lower bound of $2$ on its asymptotic competitive ratio.
\end{abstract}

\section{Introduction}

In online strip packing, items arrive one at a time without knowledge of future items and must be placed irrevocably in a strip of fixed width and unbounded height. In the model considered here, feasibility depends not only on the final geometry of the packing but also on how each item reaches its position: an arriving square must enter from above along a collision-free, monotonic downwards (never moving upwards) path and must come to rest on the strip bottom or on previously placed squares. Thus a geometrically empty section of the strip may be unusable because it is unreachable or lacks support.

The square version of this Tetris-and-gravity model was studied by Fekete, Kamphans, and Schweer~\cite{FeketeKamphansSchweer2014}. They prove that the classic \textsc{BottomLeft} algorithm has asymptotic competitive ratio $3.5$ and introduced \textsc{SlotAlgorithm}, which recursively partitions the strip into dyadic slots, rounds each square to determine candidate slots, and places it in the lowest candidate slot. They show that \textsc{SlotAlgorithm} has asymptotic competitive ratio $34/13\approx2.6154$, and establish a general lower bound of $3/2$ for deterministic online algorithms. They also show instances giving a lower bound of $2$ for \textsc{SlotAlgorithm} itself. To the best of our knowledge, the $34/13$ guarantee remain the best known upper bound for this square-packing model. Our main result improves this upper bound to $2.373318506\ldots$.

\paragraph{Our approach.}
\AsymmetricSlots{} replaces the dyadic hierarchy of \textsc{SlotAlgorithm} by an asymmetric recursive decomposition. Every slot of width $w$ is split into a \emph{wide child} of width $pw$ and a \emph{narrow child} of width $qw$, where $p\in(1/2,1)$ and $q=1-p$. At each slot, a square that does not fit the wide child is assigned to the current slot; a square that fits only the wide child is forced there; and a square that fits both children is routed to the less loaded child. Routing decisions are therefore local and are made while descending the hierarchy. See Table \ref{tab:routing-rules} a classification of these rules.

A similar analysis to the \textsc{SlotAlgorithm} for our new decomposition scheme does not yield an improve upper bound, instead we use a local charging argument and exploit the asymmetry. The asymmetric hierarchy separates placements into two useful types. Squares that are large relative to the current slot are forced, and their area can be used directly to account for the height they create. Smaller squares retain routing flexibility, allowing the algorithm to balance the two children locally. We capture the interaction between these two types of placements using a phase-based accounting argument. The key point is that the loss associated with flexible placements can be bounded locally and does not accumulate with the depth of the recursive hierarchy.

For an optimized split parameter $p^\star$, the phase analysis gives
\[
\AS_{p^\star}(\sigma)
    \le
    2.373318506\ldots\,A(\sigma)+O(1)
    \le
    2.373318506\ldots\,\OPT(\sigma)+O(1),
\]
where $A(\sigma)$ is the total area of the input squares and $\OPT(\sigma)$ is the height of an optimal offline packing. Thus \(\AS_{p^\star}\) has asymptotic competitive ratio at most $2.373318506\ldots$. The exact optimizing parameter and additive constant are derived in Section~\ref{sec:analysis}; the nearby rational split $p{:}q=11{:}8$ gives the explicit ratio $19/8=2.375$.

The conference version of Fekete et al. listed rectangles as a direction for future work, noting that they might require a different analysis~\cite{FeketeKamphansSchweer2009}. We pursue this direction when only allowing translation of rectangles of bounded aspect ratio. If every rectangle has aspect ratio at most $\kappa$, routing according to width and accounting for the rectangle height yield an asymptotic competitive ratio of $O(\kappa)$. Conversely, an alternating construction gives an $\Omega(\kappa)$ lower bound for every online algorithm in the present Tetris-and-gravity model. Thus the optimal dependence on the aspect-ratio bound is $\Theta(\kappa)$; see Section~\ref{sec:rectangles}. For the square algorithm itself, a repeated-square construction shows
\[
\CR(\AS_p)\ge2
\qquad\text{for}\qquad
\frac12<p<\frac{\sqrt5-1}{2},
\]
a range containing both $p^\star$ and $11/19$. This lower bound is algorithm-specific, just as the value $2$ is for \textsc{SlotAlgorithm}; the best-known general lower bound remains $3/2$~\cite{FeketeKamphansSchweer2014}. Consequently,
\[
2\le\CR(\AS_{p^\star})\le2.373318506\ldots,
\]
while the deterministic online problem in general currently lies between $3/2$ and $2.373318506\ldots$.

\paragraph{Related work.}
Strip packing is a classical problem in combinatorial optimization; see the surveys~\cite{LodiMartelloMonaci2002,WascherHaussnerSchumann2007,ChristensenKhanPokuttaTetali2017}. It is strongly NP-hard via bin packing~\cite{GareyJohnson1978}, and a standard reduction from \textsc{Partition} gives the $3/2$ barrier for polynomial-time absolute approximation unless $\mathrm{P}=\mathrm{NP}$~\cite{GareyJohnson1979}. Representative milestones include Sleator's $2.5$-approximation~\cite{Sleator1980}, Steinberg's absolute approximation ratio $2$~\cite{Steinberg1997}, and the asymptotic fully polynomial-time approximation scheme of Kenyon and R\'emila~\cite{KenyonRemila2000}. The current best-known general absolute approximation ratio is $5/3+\varepsilon$~\cite{HarrenJansenPradelVanStee2014}. Allowing pseudo-polynomial running time, Jansen and Rau obtained a $(5/4+\varepsilon)$-approximation~\cite{JansenRau2019}, while G\'alvez et al.\ obtained a tight $(3/2+\varepsilon)$-approximation for skewed instances~\cite{GalvezEtAl2023}. More recently, Hougardy and Zondervan obtained a $13/6$ approximation using a suitable ordering for Bottom-Left~\cite{HougardyZondervan2026}. These results concern the geometry of a final packing and do not impose the reachability and support restrictions considered here.

Ordinary online strip packing was studied by Brown, Baker, and Katseff~\cite{BrownBakerKatseff1982}, with subsequent work on shelf algorithms and connections to bin packing~\cite{BakerSchwarz1983,CsirikWoeginger1997,HanIwamaYeZhang2007,YeHanZhang2009}. Related square-specific models include online packing of squares and cubes into fixed-size bins~\cite{EpsteinVanStee2005,EpsteinMualem2023} and online square-into-square packing~\cite{FeketeHoffmann2017}. Azar and Epstein~\cite{AzarEpstein1997} introduced a Tetris-type reachability constraint for online rectangle packing without requiring support from below; they obtained a constant competitive ratio when rotations are allowed, while unrestricted nonrotatable rectangles admit no constant competitive ratio. For \textsc{BottomLeft} in the square setting, Hougardy and Zondervan~\cite{HougardyZondervan2024} gave a $10/3-\varepsilon$ lower bound for a worst ordering, and their construction also applies to the online Tetris-and-gravity version, nearly matching the $3.5$ upper bound of Fekete et al. This result is specific to \textsc{BottomLeft} and is separate from the general upper-bound question addressed here.

\paragraph{Organization.}
Section~\ref{sec:prelim} formalizes the online Tetris-and-gravity model and introduces the notation used throughout the paper. Section~\ref{sec:algorithm} presents \AsymmetricSlots{}, including the asymmetric slot hierarchy, the local routing rule, and the proof that its placements are physically feasible. Section~\ref{sec:analysis} develops the local phase invariant, derives the global competitive bound, optimizes the split parameter, and gives a lower bound for the square-packing algorithm. Finally, Section~\ref{sec:rectangles} extends the method to rectangles of bounded aspect ratio and proves matching upper and lower bounds on the dependence on the aspect ratio.

\section{Preliminaries}\label{sec:prelim}

Let $S=[0,1]\times[0,\infty)$ be the strip. At step $i$, a square $Q_i$ of side length $a_i\in(0,1]$ is revealed and must be placed irrevocably into $S$ under the Tetris-and-graviry constraints before $Q_{i+1}$ is revealed. For a finite input $\sigma=(Q_1,\ldots,Q_n)$, let $A(\sigma):=\sum_{i=1}^n a_i^2$ denote the total square area. A square with lower-left corner $(x_i,y_i)$ occupies $[x_i,x_i+a_i]\times[y_i,y_i+a_i]$. The interiors of distinct squares must be disjoint and previously placed squares may not be moved. Squares remain axis-parallel throughout an insertion. Rotations are not allowed.

An insertion satisfies the Tetris constraint if the arriving square can be translated continuously from a position above the current packing to its final position without intersecting any previously placed square. It satisfies the gravity constraint if the vertical coordinate of its lower-left corner never increases along that path and, in its final position, its bottom side has positive-length intersection with either the strip bottom or the union of top sides of earlier squares. A packing is feasible if every insertion satisfies both constraints. Its height is the largest occupied vertical coordinate.

Let $\OPT(\sigma)$ denote the minimum height of an offline packing of the squares in $\sigma$ into $S$. The offline packing knows the entire input in advance and may place the squares in any order. Since the strip has unit width,
\begin{equation}
A(\sigma)\le\OPT(\sigma).
\label{eq:area-lb}
\end{equation}

For an online algorithm $\mathcal A$, let $\mathcal A(\sigma)$ denote the height of its packing on $\sigma$. We say that $\mathcal A$ has asymptotic competitive ratio at most $\rho$ if there exists a constant $\beta\ge0$, independent of $\sigma$, such that
\[
\mathcal A(\sigma)\le\rho\OPT(\sigma)+\beta
\]
for every finite input $\sigma$. The asymptotic competitive ratio $\CR(\mathcal A)$ is the infimum of all such $\rho$.

\section{The \AsymmetricSlots{} Algorithm}\label{sec:algorithm}

Fix $p\in(1/2,1)$ and put $q:=1-p$. The hierarchy is fixed before any square arrives.


\paragraph{Slot hierarchy.}

A slot is a vertical substrip
\[
T=I(T)\times[0,\infty),
\]
where
\[
I(T)=[x(T),x(T)+w(T)]\subseteq[0,1]
\]
and $w(T)$ is the slot width. Where the root $T_0 = S$. Every slot $T$ of width $w$ is partitioned into a wide child $T_w$ and a narrow child $T_n$ with
\[
I(T_w)=[x(T),x(T)+pw],
\]
\[
I(T_n)=[x(T)+pw,x(T)+pw+qw].
\]
The split is repeated recursively. After $d$ descents, the current width is at most $p^d$.
\begin{figure}[H]
\centering
\begin{tikzpicture}[
  x=10.8cm,
  y=0.68cm,
  font=\small,
  slot label/.style={inner sep=0pt,font=\footnotesize}
]
  \draw (0,2.55) rectangle (1,2.9);
  \node[left] at (0,2.72) {root};

  \draw (0,1.5) rectangle (1,1.85);
  \draw (0.58,1.5)--(0.58,1.85);
  \node[left] at (0,1.67) {one level};
  \node[slot label] at (0.29,1.67) {$p$};
  \node[slot label] at (0.79,1.67) {$q$};

  \draw (0,0.34) rectangle (1,0.90);
  \draw (0.3364,0.34)--(0.3364,0.90);
  \draw (0.58,0.34)--(0.58,0.90);
  \draw (0.8236,0.34)--(0.8236,0.90);
  \node[left] at (0,0.62) {two levels};
  \node[slot label] at (0.1682,0.62) {$p^2$};
  \node[slot label] at (0.4582,0.62) {$pq$};
  \node[slot label] at (0.7018,0.62) {$qp$};
  \node[slot label] at (0.9118,0.62) {$q^2$};
\end{tikzpicture}
\caption{The nested slot geometry. The drawing uses approximate widths only for visualization. Analytically the widths are exactly $p$ and $q=1-p$. Candidate slots are nested regions of one physical strip, not separate bins.}
\label{fig:slots}
\end{figure}
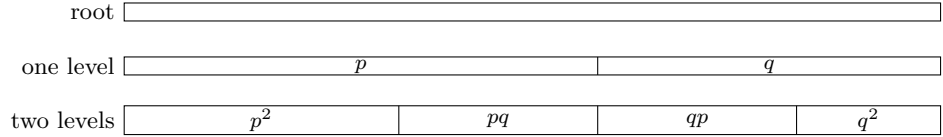

Shared slot boundaries have measure zero and play no role in area, positive-length support, or height. A square of side $a$ fits a slot $T$ if $a\le w(T)$.

\paragraph{Virtual height profile.}

The algorithm maintains a piecewise-constant virtual height profile
\[
h:[0,1]\to[0,\infty),
\]
initially $h\equiv0$. For a slot $T$, its virtual height is
\[
\tau(T):=\sup_{x\in\operatorname{int}I(T)}h(x),
\]
and the global virtual height is
\[
\Hvirt:=\sup_{x\in(0,1)}h(x).
\]
If a square of side $a$ is assigned to target slot $T$, let $b=\tau(T)$ immediately before the update. The virtual update is
\[
h(x)\leftarrow b+a\qquad x\in\operatorname{int}I(T).
\]
This full-slot virtual update is bookkeeping only. The physical square remains an $a\times a$ square.

\paragraph{Routing rule.}

Each square starts at the root. At a current slot $T$ of width $w$, exactly one of the following cases applies:

\begin{table}[H]
    \centering
    \begin{tabular}{@{}lll@{}}
        \toprule
        case & side-length range & action \\
        \midrule
        Direct   & $a > pw$           & assign to $T$ and stop \\
        Forced   & $qw < a \leq pw$   & continue in $T_w$ \\
        Flexible & $a \leq qw$        & continue in the child of smaller $\tau$ \\
        \bottomrule
    \end{tabular}
    \caption{Routing rules}
    \label{tab:routing-rules}
\end{table}

\noindent Ties in the flexible case are broken in favor of the wide child, making the algorithm deterministic.

Algorithm~\ref{alg:as-insert} gives the complete routing and insertion rule for one arriving square.
\begin{algorithm}[H]
\caption{Insertion rule for $\AS_p$}
\label{alg:as-insert}
\begin{algorithmic}[1]
\Require An arriving square $Q$ of side length $a$
\State $T\gets T_0$
\While{$a\le p\,w(T)$}
  \If{$a>q\,w(T)$}
    \State $T\gets T_w$ \Comment{forced descent}
  \ElsIf{$\tau(T_w)\le\tau(T_n)$}
    \State $T\gets T_w$ \Comment{flexible descent, wide on a tie}
  \Else
    \State $T\gets T_n$ \Comment{flexible descent}
  \EndIf
\EndWhile
\State $b\gets\tau(T)$ \Comment{$T$ is the target slot}
\State Insert the square physically in $T$
\State $h(x)\gets b+a$ for every $x\in\operatorname{int}I(T)$
\end{algorithmic}
\end{algorithm}

The descent is well defined. Whenever the algorithm continues to a child, the routing condition guarantees that the square fits that child. Since after $d$ descents the current width is at most $p^d$ and $p<1$, every square of positive side length eventually stops. Every target slot therefore satisfies
\begin{equation}
pw(T)<a\le w(T).
\label{eq:target-range}
\end{equation}

For illustration, let $p=11/19$ and $q=8/19$. Table~\ref{tab:routing} shows the three routing cases for a sequence of 8 squares, and Figure~\ref{fig:eight-square-packing} shows the resulting physical packing. Here W and N denote one descent into the wide and narrow child, respectively. Each square starts at the root. A direct assignment overwrites the virtual height profile on its target slot but never moves an earlier physical square.

\begin{center}
\begin{minipage}[t]{0.59\textwidth}
\vspace{0pt}
\captionsetup{hypcap=false}
\centering
\footnotesize
\begin{tabular}{@{}cccc@{}}
\toprule
square & side & route / target & routing summary\\
\midrule
$Q_1$ & $1/2$ & W & forced, then direct\\
$Q_2$ & $1/4$ & N & flexible, then direct\\
$Q_3$ & $1/5$ & NW & flexible, forced, then direct\\
$Q_4$ & $3/20$ & NN & flexible twice, then direct\\
$Q_5$ & $1/10$ & NNW & flexible twice, forced, then direct\\
$Q_6$ & $4/25$ & WWW & two wide tie-breaks, forced, then direct\\
$Q_7$ & $3/10$ & N & flexible, then direct\\
$Q_8$ & $3/5$ & root & direct at the root\\
\bottomrule
\end{tabular}
\captionof{table}{Routing in the eight-square example for $p=11/19$.}
\label{tab:routing}
\end{minipage}\hfill
\begin{minipage}[t]{0.36\textwidth}
\vspace{0pt}
\captionsetup{hypcap=false}
\centering
\begin{tikzpicture}[x=4.05cm,y=4.05cm,font=\scriptsize]
  \draw[densely dashed,gray!65] (0.578947,0)--(0.578947,1.4);
  \draw[densely dashed,gray!65] (0.822715,0)--(0.822715,1.4);
  \draw[densely dashed,gray!65] (0.925354,0)--(0.925354,1.4);

  \filldraw[fill=blue!16]   (0,0) rectangle (0.5,0.5);
  \node at (0.25,0.25) {$Q_1$};
  \filldraw[fill=orange!22] (0.578947,0) rectangle (0.828947,0.25);
  \node at (0.703947,0.125) {$Q_2$};
  \filldraw[fill=green!18]  (0.578947,0.25) rectangle (0.778947,0.45);
  \node at (0.678947,0.35) {$Q_3$};
  \filldraw[fill=yellow!28] (0.822715,0.25) rectangle (0.972715,0.4);
  \node at (0.897715,0.325) {$Q_4$};
  \filldraw[fill=red!14]    (0.822715,0.4) rectangle (0.922715,0.5);
  \node at (0.872715,0.45) {$Q_5$};
  \filldraw[fill=blue!24]   (0,0.5) rectangle (0.16,0.66);
  \node at (0.08,0.58) {$Q_6$};
  \filldraw[fill=red!22]    (0.578947,0.5) rectangle (0.878947,0.8);
  \node at (0.728947,0.65) {$Q_7$};
  \filldraw[fill=violet!18] (0,0.8) rectangle (0.6,1.4);
  \node at (0.3,1.1) {$Q_8$};

  \draw[thick] (0,0) rectangle (1,1.4);
  \node[below] at (0,0) {$0$};
  \node[below] at (1,0) {$1$};
\end{tikzpicture}


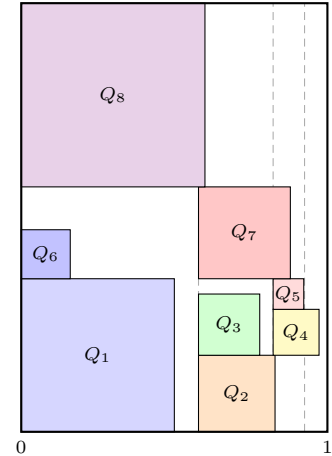
\captionof{figure}{Physical packing for the example. Dashed lines indicate relevant slot boundaries.}
\label{fig:eight-square-packing}
\end{minipage}
\end{center}


\paragraph{Physical feasibility.}
The virtual height profile controls routing but is not a physical packing: a square may fall below the virtual height of its target slot before reaching support. Suppose the current packing consists of squares $Q_1,\ldots,Q_k$,
where $Q_i$ has side length $a_i$ and lower-left corner $(x_i,y_i)$.
For every horizontal coordinate $x\in(0,1)$ that is not a hierarchy
boundary, define 
\[
g(x):=\max\Bigl(\{\,y_i+a_i:x\in(x_i,x_i+a_i)\,\}\cup\{0\}\Bigr).
\]
In words, $g(x)$ is the height at horizontal
position $x$: among all placed squares whose horizontal span contains
$x$, it records the height of the highest top edge, and it is $0$ if no
such square exists. 

We say that the virtual profile \emph{dominates} the physical packing if $g(x)\le h(x)$ for every such $x$.

\begin{lemma}\label{lem:physical}
Suppose a square of side $a$ has target slot $T$, and let $b=\tau(T)$ before the virtual update. Align the square with the left endpoint $x(T)$ of $I(T)$ above the current packing, lower it vertically to bottom coordinate $b$, and continue until it first reaches support. The resulting insertion satisfies the Tetris and gravity constraints. After setting $h(x)=b+a$ on $\operatorname{int}I(T)$, the new virtual profile dominates the new physical packing whenever the old virtual profile dominated the old packing.
\end{lemma}

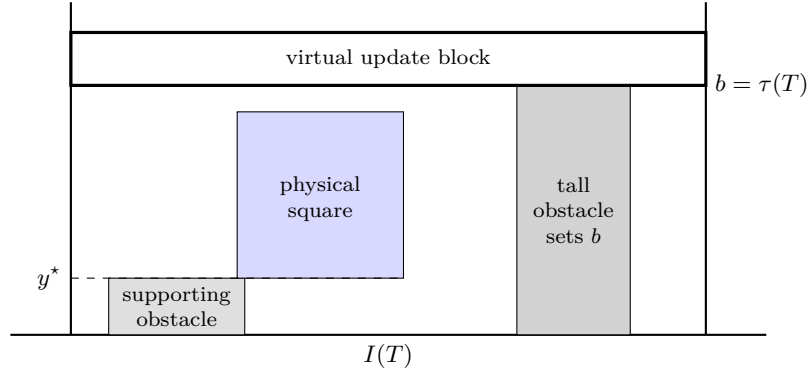
\begin{figure}[H]
\centering
\begin{tikzpicture}[x=1cm,y=1cm,font=\small,
  box label/.style={font=\footnotesize,align=center,inner sep=1pt}]
  \draw[thick] (0,0)--(10,0);
  \draw[thick] (0.8,0)--(0.8,4.4);
  \draw[thick] (9.2,0)--(9.2,4.4);
  \node[below] at (5,0) {$I(T)$};

  \filldraw[fill=gray!25] (1.3,0) rectangle (3.1,0.75);
  \node[box label,text width=1.5cm] at (2.2,0.37) {supporting\\obstacle};

  \filldraw[fill=blue!15] (3.0,0.75) rectangle (5.2,2.95);
  \node[box label,text width=1.65cm] at (4.1,1.8) {physical\\square};

  \filldraw[fill=gray!35] (6.7,0) rectangle (8.2,3.3);
  \node[box label,text width=1.2cm] at (7.45,1.65) {tall\\obstacle\\sets $b$};

  \draw[dashed] (0.8,3.3)--(9.2,3.3);
  \node[right] at (9.2,3.3) {$b=\tau(T)$};
  \draw[dashed] (0.8,0.75)--(5.2,0.75);
  \node[left] at (0.8,0.75) {$y^\star$};

  \draw[very thick] (0.8,3.3) rectangle (9.2,4.0);
  \node[box label,text width=3.6cm] at (5,3.65) {virtual update block};
\end{tikzpicture}
\caption{The virtual update spans the entire target slot from height $b$ to $b+a$, while the physical square occupies only its own footprint and can fall to a lower support height $y^\star$.}
\label{fig:physical}
\end{figure}

\begin{proof}
By~\eqref{eq:target-range}, the arriving square's horizontal span $[x(T),x(T)+a]$ lies inside $I(T)$. Let $\mathcal O$ be the set of earlier squares whose horizontal interiors overlap the arriving square. Every square in $\mathcal O$ has top coordinate at most $b$: otherwise, the positive-length overlap interval contains a point $x$ that is not a hierarchy boundary, and there $g(x)>b$, while domination gives $g(x)\le h(x)\le\tau(T)=b$, a contradiction. Hence, after the initial horizontal alignment above the packing, the new square has a collision-free vertical path down to bottom coordinate $b$.

Let $y^\star$ be the maximum of $0$ and the top coordinates of the squares in $\mathcal O$. Lower the square further to bottom coordinate $y^\star$. The path remains collision-free by definition of $y^\star$. If $y^\star>0$, a square attaining this maximum overlaps the new square over positive horizontal length and therefore supports it. If $y^\star=0$, the strip bottom supports it. The vertical coordinate never increases, so the insertion is feasible under the Tetris and gravity constraints.

The final top coordinate is at most $b+a$. On $I(T)$ the new virtual value is $b+a$. Outside $I(T)$ the virtual profile is unchanged. Thus domination is preserved.
\end{proof}

The restriction to non-boundary coordinates causes no loss: hierarchy boundaries form a countable set, while every square has positive horizontal width, so the interior of every square contains a non-boundary coordinate. Lemma~\ref{lem:physical} therefore implies that the physical packing height is at most $\Hvirt$.

The hierarchy can be maintained implicitly. Only visited slots need be stored. Each node records enough information to recover its two child heights together with a lazy ``uniform height'' value. A direct update makes the slot uniform, so its stored descendants can be discarded. Routing a square of side $a$ visits $O(1+\log(1/a))$ slots, and the same bound holds for the number of newly created nodes in that insertion. This is an operation-count bound in a unit-cost exact-arithmetic model. Bit complexity depends on the representation of the input side lengths.

For the rational split $p=11/19$, all hierarchy widths are rational products of $p$ and $q$ and comparisons can be performed exactly when the input side lengths are represented exactly. The minimizing parameter $p^\star$ is algebraic. One may either work in a real-RAM/algebraic-number model or choose a rational $p>p^\star$ arbitrarily close to $p^\star$. Continuity then gives a competitive coefficient arbitrarily close to $\rho^\star$.

\section{Analysis for Squares}\label{sec:analysis}

\subsection{The local phase invariant}\label{sec:local-phase}

\paragraph{Local accounting.}

Fix a slot of width $w$. Multiplying a virtual-height increase by $w$ gives the area of the corresponding rectangle over the slot. This motivates an accounting rule with coefficient $\rho$ against square area and an additive open-phase credit $\delta w^2$.

A square that is forced at the current slot satisfies $a>qw$ and can increase the slot maximum by at most $a$. Set
\[
\rho:=\frac1q.
\]
Then
\[
wa-\rho a^2\le0
\qquad\text{whenever }a\ge qw.
\]
so every forced square can pay for its own possible height increase.

For a flexible square, $0<a\le qw$, the concave expression $wa-\rho a^2$ is maximized at $a=qw/2$. Hence
\[
wa-\rho a^2\le\frac q4w^2.
\]
The two subslots retain credits $\delta p^2w^2$ and $\delta q^2w^2$, so the credit released at the parent scale is
\[
\delta w^2-\delta(p^2+q^2)w^2
=2\delta pq\,w^2.
\]
Thus flexible squares are covered whenever
\[
\delta\ge\frac1{8p}.
\]

A square assigned directly to the current slot has $a>pw$ and closes the current phase. Writing $r=a/w\ge p$, the function $r\mapsto\rho r^2-r$ is increasing. Such a square can therefore repay the phase credit whenever
\[
\delta\le\rho p^2-p
=\frac{p(2p-1)}q.
\]

Hence a feasible credit exists if and only if
\begin{equation}
16p^3-8p^2+p-1\ge0.
\label{eq:cubic}
\end{equation}
We use the smallest admissible value,
\[
\delta:=\frac1{8p},
\]
which minimizes the additive term that may survive at the root.

\paragraph{Phase decomposition.}

Fix a slot $T$ of width $w$ and consider only squares that, when their routing enters $T$, are ultimately assigned to $T$ or to one of its descendants. At the start of a local $T$-phase, the virtual profile is constant on $I(T)$. Let this common baseline be $b$.

An open $T$-phase contains no square whose target is exactly $T$. A completed $T$-phase ends with the first square whose target is $T$. That direct square performs a full-slot virtual update and closes the phase. For an open phase define
\[
M:=\sup_{x\in\operatorname{int}I(T)}h(x)-b,
\]
and let $A$ be the sum of physical square areas in the phase. For a completed phase let $\Delta$ be the final uniform virtual rise above $b$, including the closing square.

For an open phase, we call $wM$ the \emph{slot-height product}. Figure~\ref{fig:phases}(a) shows it on the parent slot. The rectangle begins at the phase baseline $b$, spans the full slot width $w$, and reaches the current level $\tau(T)=b+M$. Panel~(b) shows the same child-phase history followed by the full-slot virtual update caused by the closing direct square of side $a$. The child phases labeled ``final open'' are the open child phases present immediately before that parent-closing update. The physical square itself is not drawn, since its gravity-feasible placement was established in Lemma~\ref{lem:physical}. The figure is schematic: the rectangular blocks represent virtual-height contributions and phase structure, not physical rectangles in the packing.

\begin{figure}[H]
\centering
\begin{tikzpicture}[x=1cm,y=1cm,font=\small]
  \begin{scope}[shift={(0,0)}]
    \fill[gray!8] (0,0) rectangle (5.2,2.45);
    \draw[thick] (0,0) rectangle (5.2,2.45);
    \draw[dashed] (3.05,0)--(3.05,2.45);

    \filldraw[fill=gray!20] (0.22,0) rectangle (2.83,0.58);
    \node at (1.53,0.29) {completed};
    \filldraw[fill=gray!20] (0.22,0.58) rectangle (2.83,1.12);
    \node at (1.53,0.85) {completed};
    \filldraw[fill=white] (0.22,1.12) rectangle (2.83,2.45);
    \node at (1.53,1.79) {open};

    \filldraw[fill=gray!20] (3.27,0) rectangle (4.98,0.88);
    \node at (4.13,0.44) {completed};
    \filldraw[fill=white] (3.27,0.88) rectangle (4.98,1.82);
    \node at (4.13,1.35) {open};

    \node[fill=gray!8,inner sep=1.5pt,font=\scriptsize] at (4.02,2.20)
      {area $wM$};

    \draw[decorate,decoration={brace,mirror,amplitude=4pt}]
      (0,-0.38)--(5.2,-0.38)
      node[midway,below=5pt] {$w$};
    \draw[<->] (-0.28,0)--(-0.28,2.45);
    \node[left] at (-0.32,1.225) {$M$};
    \node[right] at (5.25,0) {$b$};
    \node[right] at (5.25,2.45) {$\tau(T)$};

    \node[below=2pt] at (1.53,0) {wide child};
    \node[below=2pt] at (4.13,0) {narrow child};
    \node[font=\small] at (2.60,-1.02) {(a) open $T$-phase};
  \end{scope}

  \begin{scope}[shift={(6.8,0)}]
    \draw[thick] (0,0) rectangle (5.2,3.25);
    \draw[dashed] (3.05,0)--(3.05,2.45);

    \filldraw[fill=gray!20] (0.22,0) rectangle (2.83,0.58);
    \node at (1.53,0.29) {completed};
    \filldraw[fill=gray!20] (0.22,0.58) rectangle (2.83,1.12);
    \node at (1.53,0.85) {completed};
    \filldraw[fill=white] (0.22,1.12) rectangle (2.83,2.45);
    \node at (1.53,1.79) {final open};

    \filldraw[fill=gray!20] (3.27,0) rectangle (4.98,0.88);
    \node at (4.13,0.44) {completed};
    \filldraw[fill=white] (3.27,0.88) rectangle (4.98,1.82);
    \node at (4.13,1.35) {final open};

    \filldraw[fill=gray!10,very thick] (0,2.45) rectangle (5.2,3.25);
    \node[font=\scriptsize] at (2.60,2.85) {full-slot virtual update};

    \node[left] at (-0.08,0) {$b$};
    \node[anchor=west] at (5.30,2.45) {$\tau(T)$};
    \node[anchor=west] at (5.30,3.25) {$b+\Delta$};

    \draw[<->] (6.35,2.45)--(6.35,3.25);
    \node[anchor=west] at (6.46,2.85) {$a$};

    \draw[<->] (7.25,0)--(7.25,3.25);
    \node[anchor=west] at (7.38,1.625) {$\Delta$};

    \node[below=2pt] at (1.53,0) {wide child};
    \node[below=2pt] at (4.13,0) {narrow child};
    \node[font=\small] at (2.60,-1.02) {(b) completed $T$-phase};
  \end{scope}
\end{tikzpicture}
\caption{The two forms of a local $T$-phase. The figure is schematic: the rectangular blocks represent virtual-height contributions and phase structure, not physical rectangles in the packing. (a) In an open phase, the rectangle from $b$ to $\tau(T)=b+M$ spans the slot width $w$ and has area $wM$, the slot-height product. The child histories consist of completed phases followed by at most one open phase. (b) The labels ``final open'' mark the open child phases present immediately before a direct square of side $a$ closes the parent phase. Its full-slot virtual update raises the level from $\tau(T)=b+M$ to $b+\Delta$, where $\Delta=M+a$. The band represents this virtual update, not the physical square.}
\label{fig:phases}
\end{figure}
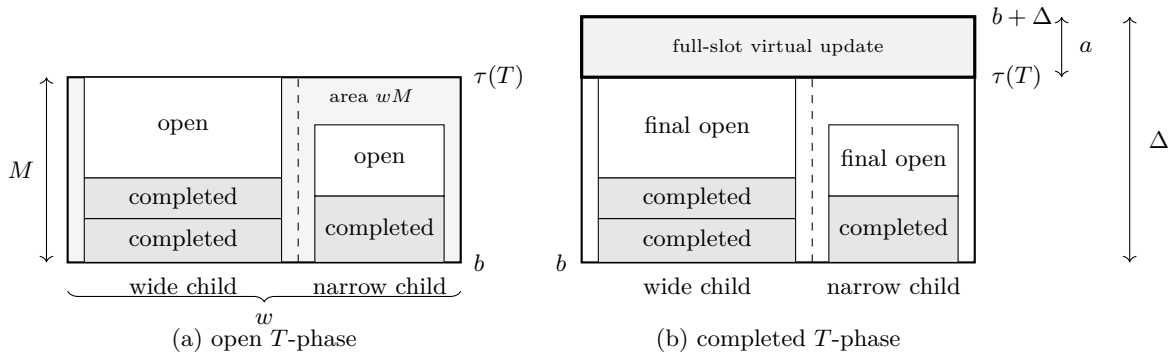

Within either child of $T$, the history during one parent phase consists of zero or more completed child phases, followed by at most one open child phase. Every completed child phase ends with a full-child update, which supplies the flat baseline for the next child phase.

\begin{lemma}\label{lem:phase}
Assume~\eqref{eq:cubic}, and let $\rho=1/q$ and $\delta=1/(8p)$. For every slot $T$ of width $w$,
\begin{align}
wM&\le\rho A+\delta w^2 &&\text{for every open $T$-phase},\label{eq:open-phase}\\
w\Delta&\le\rho A &&\text{for every completed $T$-phase}.\label{eq:closed-phase}
\end{align}
\end{lemma}

\begin{proof}
We use simultaneous strong induction on the number of squares in the local phase. The empty open phase satisfies~\eqref{eq:open-phase}. Assume both statements hold for all local phases with fewer than $n$ squares.

\paragraph{Open phase.}
Consider a nonempty open $T$-phase with $n$ squares. Choose the last square $J$ whose insertion strictly increases the maximum virtual rise over $I(T)$. Squares arriving later add area without increasing $M$, so it is enough to prove~\eqref{eq:open-phase} immediately after $J$. Let $a$ be the side length of $J$.

If $J$ is forced at $T$, then $a>qw$. Because the $T$-phase is open, $J$ is not assigned directly to $T$, so $a\le pw$ and the first routing step is into the wide child. Regardless of how far $J$ subsequently descends, its virtual update can increase the maximum over $I(T)$ by at most $a$. If $M_{\mathrm{before}}$ and $A_{\mathrm{before}}$ describe the prefix before $J$, the induction hypothesis gives
\[
wM_{\mathrm{before}}\le\rho A_{\mathrm{before}}+\delta w^2.
\]
Hence, since $wa-\rho a^2\le0$ for $a\ge qw$,
\[
wM_{\mathrm{after}}-\rho(A_{\mathrm{before}}+a^2)
\le\delta w^2+wa-\rho a^2
\le\delta w^2.
\]

Now suppose $J$ is flexible at $T$, so $a\le qw$. Immediately before $J$, let $C$ be the selected child and $D$ its sibling. Write
\[
w_C=\lambda w,\qquad w_D=(1-\lambda)w,\qquad\lambda\in\{p,q\}.
\]
Let $m_C$ and $m_D$ be the maximum rises above the same parent baseline $b$ in the two children, and let $A_C,A_D$ be the areas of the prefix squares routed into those children. Every prefix square is routed into exactly one child, so
\[
A_{\mathrm{before}}=A_C+A_D.
\]
Since $J$ is routed to the child of smaller virtual height,
\[
m_C\le m_D.
\]

Before $J$, child $C$ contains completed phases with rises $\Delta_1,\ldots,\Delta_k$ and areas $A_1,\ldots,A_k$, followed by at most one open phase with rise $M'$ and area $A'$. All these phases contain fewer than $n$ squares, so the induction hypotheses apply. Their rises add because each completed child phase leaves the child flat. Therefore
\[
w_C m_C
=w_C\left(\sum_{i=1}^k\Delta_i+M'\right)
\le\rho\sum_{i=1}^k A_i+\rho A'+\delta w_C^2
=\rho A_C+\delta w_C^2.
\]
If there is no final open child phase, take $M'=A'=0$ and omit the credit. The same argument for $D$ gives
\begin{align*}
\lambda wm_C&\le\rho A_C+\delta\lambda^2w^2,\\
(1-\lambda)wm_D&\le\rho A_D+\delta(1-\lambda)^2w^2.
\end{align*}

The insertion of $J$ strictly raises the parent maximum. If the selected child $C$ remained below the sibling maximum $m_D$ after the insertion, then the parent maximum would still be $m_D$ and would not have increased, contradicting the choice of $J$. Thus the new parent maximum is attained in $C$. Since the selected child's maximum increases by at most $a$,
\[
M_{\mathrm{after}}\le m_C+a.
\]
Using these child-phase bounds and $m_C\le m_D$,
\begin{align*}
wM_{\mathrm{after}}-\rho(A_C+A_D+a^2)
&\le w(m_C+a)-\lambda wm_C-(1-\lambda)wm_D\\
&\quad+\delta\bigl(\lambda^2+(1-\lambda)^2\bigr)w^2-\rho a^2\\
&=(1-\lambda)w(m_C-m_D)+wa-\rho a^2+\delta(p^2+q^2)w^2\\
&\le wa-\rho a^2+\delta(p^2+q^2)w^2\\
&\le\delta w^2.
\end{align*}
The last step uses $\delta\ge1/(8p)$. With our choice $\delta=1/(8p)$ it holds with equality. This proves~\eqref{eq:open-phase}.

\paragraph{Completed phase.}
Let the closing square have side $a>pw$. Let $M$ and $A_0$ describe the preceding open part. By the induction hypothesis,
\[
wM\le\rho A_0+\delta w^2.
\]
The closing square performs a full-slot update, so $\Delta=M+a$ and
\[
w\Delta\le\rho A_0+\delta w^2+wa.
\]
For $r=a/w\ge p$, the function $r\mapsto\rho r^2-r$ is increasing. The cubic condition~\eqref{eq:cubic} guarantees that our choice $\delta=1/(8p)$ satisfies the direct-square requirement $\delta\le\rho p^2-p$. Hence
\[
\delta\le\rho p^2-p\le\rho r^2-r.
\]
Thus
\[
\delta w^2+wa\le\rho a^2,
\]
and therefore
\[
w\Delta\le\rho(A_0+a^2).
\]
This proves~\eqref{eq:closed-phase} and closes the induction.
\end{proof}

\subsection{Global bound and parameter optimization}\label{sec:global-bound}

A square whose target is the root closes the current root phase and leaves the entire virtual profile flat. Hence the execution is a sequence of completed root phases followed by at most one open root phase.

\begin{theorem}\label{thm:parameterized}
Let $p\in(1/2,1)$ satisfy~\eqref{eq:cubic} and let $q=1-p$. Then every finite input sequence $\sigma$ satisfies
\[
\AS_p(\sigma)\le\frac1qA(\sigma)+\frac1{8p}.
\]
Consequently,
\[
\AS_p(\sigma)\le\frac1q\OPT(\sigma)+\frac1{8p}.
\]
In particular, the asymptotic competitive ratio of $\AS_p$ is at most $1/q$.
\end{theorem}

\begin{proof}
Apply~\eqref{eq:closed-phase} to every completed root phase and~\eqref{eq:open-phase} to the final open root phase, whose width is $1$. The phase areas partition the input area, so
\[
\Hvirt\le\rho A(\sigma)+\delta=\frac1qA(\sigma)+\frac1{8p}.
\]
Lemma~\ref{lem:physical} bounds the physical packing height by $\Hvirt$, proving the area bound. The competitive bound then follows from~\eqref{eq:area-lb}.
\end{proof}

It remains to choose the split. The coefficient $1/(1-p)$ in Theorem~\ref{thm:parameterized} increases with $p$, so we choose the smallest $p$ satisfying~\eqref{eq:cubic}.

\begin{corollary}\label{cor:optimized}
Let $p^\star$ be the unique root in $(1/2,1)$ of
\[
16p^3-8p^2+p-1=0.
\]
Then
\[
p^\star=0.578649053069\ldots,\qquad
\rho^\star=\frac1{1-p^\star}=2.373318506305\ldots,
\]
and
\[
\AS_{p^\star}(\sigma)\le\rho^\star\OPT(\sigma)+0.216020400167\ldots.
\]
Equivalently, $\rho^\star$ is the unique real root in $(2,3)$ of
\[
8\rho^3-33\rho^2+40\rho-16=0.
\]
\end{corollary}

\begin{proof}
Let $f(p)=16p^3-8p^2+p-1$. We have $f(1/2)=-1/2$ and $f(1)=8$. Moreover,
\[
f'(p)=48p^2-16p+1>0
\]
for $p\in[1/2,1]$, since the two roots of $f'$ are $1/12$ and $1/4$. Thus $f$ has exactly one root in $(1/2,1)$ and~\eqref{eq:cubic} holds precisely for $p\ge p^\star$. Since $1/(1-p)$ is increasing, the coefficient of the fixed-split bound is minimized at $p^\star$. At the root of $f$, the flexible and direct credit constraints are both tight, so
\[
\frac1{8p^\star}=\frac{p^\star(2p^\star-1)}{1-p^\star}=0.216020400167\ldots.
\]
Replacing $p$ by $1-1/\rho$ in $f(p)=0$ gives the cubic in $\rho$.
\end{proof}

For an exact rational implementation, the nearby split $11{:}8$ gives the following simple guarantee.

\begin{corollary}\label{cor:rational}
For $p=11/19$ and $q=8/19$,
\[
\AS_{11/19}(\sigma)\le\frac{19}{8}\OPT(\sigma)+\frac{19}{88}.
\]
Hence $\AS_{11/19}$ has asymptotic competitive ratio at most $19/8=2.375$.
\end{corollary}

\begin{proof}
A direct calculation gives
\[
16\left(\frac{11}{19}\right)^3-8\left(\frac{11}{19}\right)^2+\frac{11}{19}-1
=\frac{16}{6859}>0.
\]
Thus Theorem~\ref{thm:parameterized} applies. Also
\[
\frac1q=\frac{19}{8},\qquad \frac1{8p}=\frac{19}{88}.
\]
\end{proof}

Because the rationals are dense and both $p\mapsto16p^3-8p^2+p-1$ and $p\mapsto1/(1-p)$ are continuous, for every $\varepsilon>0$ there exists a rational $\widehat p>p^\star$ satisfying~\eqref{eq:cubic} such that $\AS_{\widehat p}$ has asymptotic competitive ratio at most $\rho^\star+\varepsilon$.

\paragraph{A lower bound for \AsymmetricSlots{}.}
The square upper bound is not known to be tight. The following repeated-square construction gives a lower bound of $2$ throughout a parameter range containing both $p^\star$ and $11/19$.

\begin{proposition}\label{prop:lower}
Let
\[
\frac12<p<\frac{\sqrt5-1}{2},\qquad q=1-p.
\]
Then the asymptotic competitive ratio of $\AS_p$ is at least $2$. In particular,
\[
\CR(\AS_{p^\star})\ge2
\qquad\text{and}\qquad
\CR(\AS_{11/19})\ge2.
\]
\end{proposition}

\begin{proof}
The assumption $p<(\sqrt5-1)/2$ is equivalent to $q>p^2$. Choose
\[
a=q+\eta,\qquad0<\eta<p-\frac12.
\]
Then $q<a<1/2<p$, so at the root the square fits the wide child but not the narrow child and is forced into the wide child. Since $a>q>p^2$, it is assigned directly to that wide child at the next routing step. Every repeated copy therefore has the same target slot. Because the physical square is aligned with the target slot's left endpoint, the copies stack in one column and $\AS_p$ produces height $na$ on $n$ copies.

On the other hand, $2a<1$, so an offline packing can alternate the squares between two side-by-side columns and attain height $\lceil n/2\rceil a$. This packing is itself feasible under the Tetris and gravity constraints in arrival order. The ratio therefore tends to $2$ as $n\to\infty$.
\end{proof}

The general lower bound of Fekete, Kamphans, and Schweer is $3/2$~\cite{FeketeKamphansSchweer2014}. Proposition~\ref{prop:lower} is algorithm-specific and shows that $\CR(\AS_p)\ge2$ throughout the stated parameter range. Combining it with Corollary~\ref{cor:optimized} gives
\[
2\le\CR(\AS_{p^\star})\le2.373318506\ldots.
\]
For the general deterministic square-packing problem, the gap between $3/2$ and the upper bound proved here remains substantial.

For squares, the flexible-square proof uses only the sign of the child-height difference and discards the magnitude of the resulting balancing gain. An imbalance-dependent phase potential could retain this information. Other directions include non-self-similar splits, more than two subslots, and analyses based on lower bounds stronger than total area.

\section{Rectangles of Bounded Aspect Ratio}\label{sec:rectangles}

\paragraph{Upper bound.}
Having established the square bound, we next ask how much of the argument depends on the equality of width and height. The same framework in fact extends to rectangles of bounded aspect ratio. Fix $\kappa\ge1$, and let $\sigma=(R_1,\ldots,R_n)$ be an input of axis-parallel rectangles. Rectangle $R_i$ has width $u_i\in(0,1]$ and height $v_i>0$, with
\[
\frac{u_i}{\kappa}\le v_i\le \kappa u_i.
\]
Rotations are not allowed. For this section, let $A(\sigma):=\sum_i u_i v_i$ denote the total rectangle area, and let $\OPT(\sigma)$ denote the offline optimum height for the rectangle input.

We use the same slot hierarchy and route each rectangle according to its width. If the target of a rectangle of width $u$ and height $v$ is a slot $T$, the virtual update raises the whole slot by $v$. The physical rectangle is aligned with the left endpoint of $T$ and lowered to support exactly as in Subsection~\ref{sec:physical}. We continue to write $\AS_p(\sigma)$ for the height produced by this rectangle extension. The proof of Lemma~\ref{lem:physical} is unchanged after replacing the square side length by $u$ in the horizontal argument and by $v$ in the vertical update.

The phase induction is unchanged as well. Only the three local accounting inequalities need to be modified. Let $\varrho\ge1/q$ be a candidate competitive coefficient, and let $\delta_{\mathrm R}w^2$ be the open-phase credit. A forced rectangle has $u>qw$, and hence
\[
wv-\varrho uv=v(w-\varrho u)\le0.
\]
For a flexible rectangle, $u\le qw$. If $w-\varrho u>0$, then the aspect-ratio bound gives
\[
wv-\varrho uv\le \kappa u(w-\varrho u)\le \frac{\kappa}{4\varrho}w^2,
\]
where the last expression is the maximum of the quadratic in $u$. If $w-\varrho u\le0$, there is no shortfall. Thus the flexible case is covered whenever
\[
\delta_{\mathrm R}\ge \frac{\kappa}{8\varrho pq}.
\]
Finally, a direct rectangle has $u>pw$. Since $\varrho\ge1/q$ and $p>q$, we have $\varrho p>1$, so $r\mapsto r(\varrho r-1)$ is increasing for $r\ge p$. Together with $v\ge u/\kappa$, the closing rectangle repays the phase credit whenever
\[
\delta_{\mathrm R}\le \frac{p(\varrho p-1)}{\kappa}.
\]
The two requirements are compatible exactly when
\begin{equation}
\kappa^2\le 8\varrho p^2q(\varrho p-1).
\label{eq:rect-condition}
\end{equation}

\begin{theorem}\label{thm:rectangles}
Fix $\kappa\ge1$ and $p\in(1/2,1)$, with $q=1-p$. If
\[
\varrho\ge\frac1q
\qquad\text{and}\qquad
\kappa^2\le 8\varrho p^2q(\varrho p-1),
\]
then the rectangle version of \AsymmetricSlots{} satisfies
\[
\AS_p(\sigma)\le \varrho A(\sigma)+\frac{\kappa}{8\varrho pq}
\le \varrho\OPT(\sigma)+\frac{\kappa}{8\varrho pq}.
\]
In particular, for every fixed $\kappa$ its asymptotic competitive ratio is at most $\varrho$.
\end{theorem}

\begin{proof}
Choose $\delta_{\mathrm R}=\kappa/(8\varrho pq)$. The compatibility condition~\eqref{eq:rect-condition} ensures that this choice satisfies both the flexible and direct credit inequalities. Repeating the proof of Lemma~\ref{lem:phase}, with rectangle area $uv$ and virtual height increment $v$ in place of $a^2$ and $a$, gives the same open- and completed-phase bounds. Summing the root phases and using $A(\sigma)\le\OPT(\sigma)$ proves the claim.
\end{proof}

For fixed $p$, the smallest coefficient certified by Theorem~\ref{thm:rectangles} is
\begin{equation}
\rho_\kappa(p)
=\max\left\{
\frac1{1-p},
\frac{1+\sqrt{1+\kappa^2/(2p(1-p))}}{2p}
\right\}.
\label{eq:rect-rho}
\end{equation}
When $\kappa=1$, optimizing~\eqref{eq:rect-rho} recovers Corollary~\ref{cor:optimized}. For a simple bound valid for every $\kappa$, take $p=3/4$. Then
\[
\rho_\kappa(3/4)
=\max\left\{
4,
\frac23\left(1+\sqrt{1+\frac83\kappa^2}\right)
\right\}
=O(\kappa).
\]
Thus bounded aspect ratio is sufficient for a constant competitive ratio. The lower bound below gives a matching $\Omega(\kappa)$ dependence. In particular, the order of growth in~\eqref{eq:rect-rho} is best possible, although the constant factor may be improvable.

\paragraph{Lower bound.}
Azar and Epstein already proved that unrestricted nonrotatable rectangles admit no constant competitive ratio in their Tetris model~\cite{AzarEpstein1997}. The following proposition complements this by giving a linear lower bound in the aspect-ratio bound for the present Tetris-and-gravity model, matching Theorem~\ref{thm:rectangles} up to constants. Figure~\ref{fig:rectangle-lower-bound} illustrates the construction for $m=4$: online, the alternating sequence is forced into a vertical chain, whereas offline the thin vertical rectangles can be packed side by side and the wide shallow rectangles stacked above them.

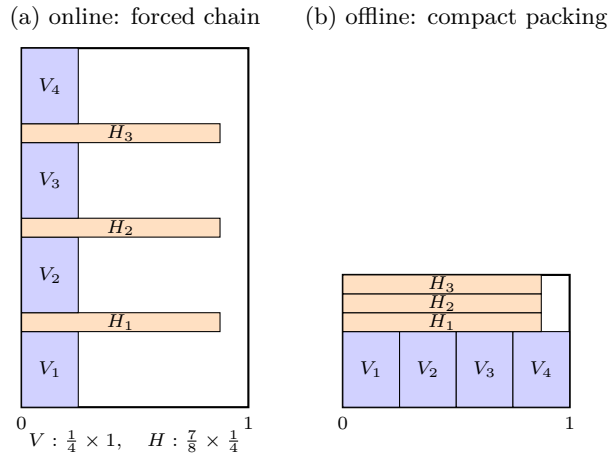
\begin{figure}[H]
  \centering
  \begin{tikzpicture}[x=3.0cm,y=1.0cm,font=\scriptsize]
  \begin{scope}
    \node[font=\small] at (0.5,5.15) {(a) online: forced chain};
    \draw[thick] (0,0) rectangle (1,4.75);

    \filldraw[fill=blue!18]   (0.00,0.00) rectangle (0.25,1.00);
    \node at (0.125,0.50) {$V_1$};
    \filldraw[fill=orange!24] (0.00,1.00) rectangle (0.875,1.25);
    \node at (0.4375,1.125) {$H_1$};

    \filldraw[fill=blue!18]   (0.00,1.25) rectangle (0.25,2.25);
    \node at (0.125,1.75) {$V_2$};
    \filldraw[fill=orange!24] (0.00,2.25) rectangle (0.875,2.50);
    \node at (0.4375,2.375) {$H_2$};

    \filldraw[fill=blue!18]   (0.00,2.50) rectangle (0.25,3.50);
    \node at (0.125,3.00) {$V_3$};
    \filldraw[fill=orange!24] (0.00,3.50) rectangle (0.875,3.75);
    \node at (0.4375,3.625) {$H_3$};

    \filldraw[fill=blue!18]   (0.00,3.75) rectangle (0.25,4.75);
    \node at (0.125,4.25) {$V_4$};

    \node[below] at (0,0) {$0$};
    \node[below] at (1,0) {$1$};
    \node[align=center] at (0.5,-0.45) {$V: \frac14\times 1,\quad H: \frac78\times \frac14$};
  \end{scope}

  \begin{scope}[xshift=4.25cm]
    \node[font=\small] at (0.5,5.15) {(b) offline: compact packing};
    \draw[thick] (0,0) rectangle (1,1.75);

    \foreach \i/\lab in {0/V_1,1/V_2,2/V_3,3/V_4} {
      \filldraw[fill=blue!18] (0.25*\i,0.00) rectangle (0.25+0.25*\i,1.00);
      \node at (0.125+0.25*\i,0.50) {$\lab$};
    }

    \filldraw[fill=orange!24] (0.00,1.00) rectangle (0.875,1.25);
    \node at (0.4375,1.125) {$H_1$};
    \filldraw[fill=orange!24] (0.00,1.25) rectangle (0.875,1.50);
    \node at (0.4375,1.375) {$H_2$};
    \filldraw[fill=orange!24] (0.00,1.50) rectangle (0.875,1.75);
    \node at (0.4375,1.625) {$H_3$};

    \node[below] at (0,0) {$0$};
    \node[below] at (1,0) {$1$};
  \end{scope}
\end{tikzpicture}
  \caption{Illustration of the rectangle lower-bound construction for $m=4$. Left: each consecutive pair has total width $1+1/(2m)>1$, so the alternating input $V_1,H_1,V_2,H_2,V_3,H_3,V_4$ is forced to preserve its vertical order online. Right: offline, the four $V$-rectangles fit side by side, while the three $H$-rectangles can be stacked above them. The proof repeats this pattern with arbitrarily many rectangles.}
  \label{fig:rectangle-lower-bound}
\end{figure}

\begin{proposition}\label{prop:rectangle-lower}
For every integer $m\ge2$, every online algorithm for nonrotatable rectangles under the Tetris and gravity constraints has asymptotic competitive ratio at least
\[
\frac{m+1}{2}
\]
on instances whose rectangles have aspect ratio at most $m$. Consequently, for rectangles of aspect ratio at most $\kappa\ge2$, the optimal competitive ratio is at least $\kappa/2$.
\end{proposition}

\begin{proof}
Let $V$ be a rectangle of width $1/m$ and height $1$, and let $H$ be a rectangle of width $1-1/(2m)$ and height $1/m$. Their aspect ratios are $m$ and $m-1/2$, respectively. For an integer $N$, let
\[
\sigma_N=(V_1,H_1,V_2,H_2,\ldots,V_{N-1},H_{N-1},V_N)
\]
be the alternating input sequence.
The widths of every two consecutive rectangles sum to
\[
\frac1m+1-\frac1{2m}=1+\frac1{2m}>1.
\]
Hence their horizontal interiors overlap for every pair of horizontal positions in the unit-width strip. A later rectangle therefore cannot pass below its predecessor along a collision-free path from above: changing their vertical order would force a moment at which their vertical interiors overlap, while their horizontal interiors necessarily overlap as well. Thus each rectangle must lie above its predecessor. For any online algorithm $\mathcal A$,
\[
\mathcal A(\sigma_N)\ge N+\frac{N-1}{m}.
\]

Offline, the $N$ copies of $V$ can be packed in rows of $m$ rectangles, and the $N-1$ copies of $H$ can be stacked above them. Hence
\[
\OPT(\sigma_N)\le \left\lceil\frac Nm\right\rceil+\frac{N-1}{m}.
\]
Taking $N$ through multiples of $m$ and letting $N\to\infty$ gives
\[
\CR(\mathcal A)\ge
\lim_{N\to\infty}
\frac{N+(N-1)/m}{N/m+(N-1)/m}
=\frac{m+1}{2}.
\]
For a real aspect-ratio bound $\kappa\ge2$, take $m=\lfloor\kappa\rfloor$. Since $m+1>\kappa$, the lower bound is greater than $\kappa/2$.
\end{proof}

Together with Theorem~\ref{thm:rectangles}, Proposition~\ref{prop:rectangle-lower} shows that the optimal dependence on the maximum aspect ratio is $\Theta(\kappa)$. Letting $m\to\infty$ also recovers the fact that unrestricted nonrotatable rectangles admit no uniform constant competitive ratio.

For bounded-aspect-ratio rectangles, the order of growth is now settled. Improving the constants in the upper and lower bounds remains open.

\section*{Use of generative AI.}
OpenAI's GPT-5.6 Sol was used iteratively to assist with language and presentation, identifying potentially relevant literature, exploring candidate proof strategies, and generating code for packing simulations and a Lean formalization. The author, who takes full responsibility for the paper.

\bibliographystyle{abbrv}
\bibliography{references}

\end{document}